\documentclass{article}
\RequirePackage[top=1.4in,bottom=1.4in,left=2in,right=2in]{geometry}
\RequirePackage[labelfont=bf]{caption}

\RequirePackage{amsmath,amsfonts,amssymb,amsthm}
\RequirePackage{graphicx}
\RequirePackage{hyperref}
\RequirePackage{framed}

\newtheorem{claim}{Claim}
\newtheorem{definition}{Definition}
\newtheorem*{notation}{Notation}
\newtheorem*{note}{Note}

\DeclareMathOperator{\range}{range}
\DeclareMathOperator{\diagstar}{diag*}
\newcommand{\Reals}{\mathbb{R}}
\newcommand{\T}{\mathsf{T}}

\newcommand{\vx}{\ensuremath{\vec{x}}}
\newcommand{\vh}{\ensuremath{\vec{h}}}
\newcommand{\vb}{\ensuremath{\vec{b}}}

\newcommand{\U}{\ensuremath{\mathcal{U}}}
\newcommand{\V}{\ensuremath{\mathcal{V}}}
\newcommand{\B}{\ensuremath{\mathcal{B}}}

\begin{document}

    \begin{center}
        \Large
        Spectral Learning of Hidden Markov Model \\ Predictive State Representations\\
        \normalsize
        Matthew James Johnson, MIT\\
        \footnotesize
        Revised \today
    \end{center}

    This document is my summary explanation of the algorithm in
    \href{http://www.cs.mcgill.ca/~colt2009/papers/011.pdf}{``A Spectral
        Algorithm for Learning Hidden Markov Models''} (COLT 2009), though
    there may be some slight notational inconsistencies with the original
    paper. The exposition and the math here are quite different, so if you
    don't like this explanation, try the original paper!

    The idea is to maintain output predictions in a recursive inference
    algorithm, instead of the usual method of maintaining hidden state
    predictions, and to represent the HMM only in terms of the maps necessary
    to update output predictions given new data. This approach limits the
    inference computations the algorithm can perform (it can't answer any
    queries about the hidden states since it doesn't explicitly deal with them
    at all), but it also reduces the complexity of the model parameters that
    are learned and thus makes learning easier. The learning algorithm uses an
    SVD and matrix operations, so it avoids the local-optima problems of EM or
    any other algorithms based on maximizing data likelihood over the usual HMM
    parameterization. The COLT paper includes error bounds and analysis.

    \begin{notation}
        For a vector $v$ in a subspace $\V \subseteq \Reals^k$ and a matrix $C$
        with linearly independent columns and $\range(C) \supseteq \V$ I will
        use $[v]^C$ to denote the coordinate vector of $v$ relative to the
        ordered basis  given by the columns of $C$, and $[v]$ or simply
        $v$ to denote the coordinate vector of $v$ relative to the standard
        basis of $\Reals^k$.  Similarly, for a linear map $\mathcal{A}: \V \to
        \V$ I will use $[\mathcal{A}]_C^C$ to denote the matrix of
        $\mathcal{A}$ relative to domain and codomain bases given by the
        columns of $C$, and $[\mathcal{A}]$ to indicate its matrix relative to
        standard bases. For a matrix $A$ I will also use $[A]_{ij}$ to denote
        the $(i,j)$th entry.
    \end{notation}

    \begin{figure}[ht]
        \includegraphics[width=2.5in]{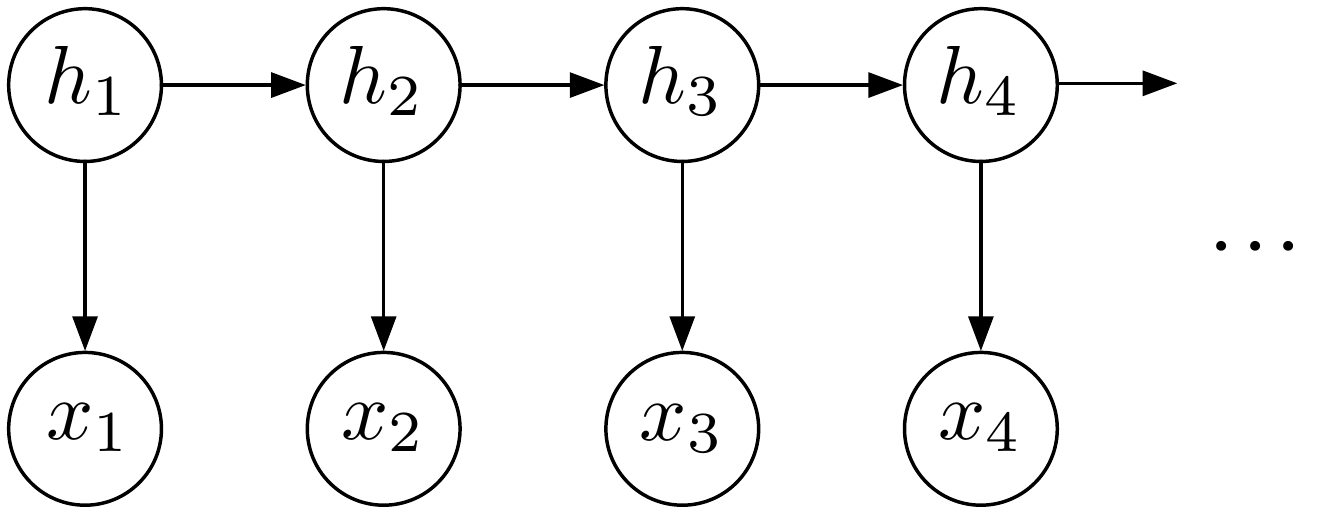}
        \centering
    \end{figure}

    \begin{definition}[Hidden Markov Model]
        A time-homogeneous, discrete \emph{Hidden Markov Model (HMM)} is a
        probability distribution on random variables $\{(x_t,h_t)\}_{t \in
            \mathbb{N}}$ satisfying the conditional independences implied by the
        graphical model, where $\range(h_t) = [m] := \{1,2,\ldots,m\}$ and
        $\range(x_t) = [n]$ where $n \geq m$. The \emph{standard
            parameterization} is the triple $(T,O,\pi)$, where
        \begin{align*}
            T \in \Reals^{m \times m}, \quad [T]_{ij} &= \Pr[h_t = i | h_{t-1} = j] \\
            O \in \Reals^{n \times m}, \quad [O]_{ij} &= \Pr[x_t = i | h_t = j] \\
            \pi \in \Reals^m, \quad [\pi]_j &= \Pr[h_1 = j].\\
        \end{align*}
        We will assume $T$ and $O$ to have full column rank and $[\pi]_j > 0 \;
        \forall j \in [m]$.
    \end{definition}

    \begin{definition}[Observation Prediction]
        An \emph{observation prediction} for any time $t$ is a vector
        $\vx_{t} \in \Reals^n$ defined in the standard basis by
        \begin{align}
            [\vx_t]_i := \Pr[x_t = i | x_{1:t-1} = \bar{x}_{1:t-1}]
        \end{align}
        for some fixed (implicit) sequence $\bar{x}_{1:t-1}$.
    \end{definition}

    \begin{claim}
        Every observation prediction $\vx_t$ lies in a subspace $\U :=
        \range(O) \subseteq \Reals^n$ with $\dim(\U) = m$.
    \end{claim}
    \begin{proof}
        By the conditional independences of the HMM, for any $t$ we have
        \begin{align}
            \Pr[x_t = i | x_{1:t-1} = \bar{x}_{1:t-1}] &= \sum_{j \in [m]} \Pr[x_t = i | h_t = j] \nonumber \\
            &\quad \cdot \Pr[h_t = j | x_{1:t-1} = \bar{x}_{1:t-1}]
        \end{align}
        so therefore we can write
        \begin{align}
            \vx_{t} = O \vh_{t}, \quad [\vh_t]_j := \Pr[h_t = j | x_{1:t-1} =
            \bar{x}_{1:t-1}].
        \end{align}
        Equivalently, we can say $[\vx_t]^O = \vh_t$. See Figure~\ref{fig:claim1}.
    \end{proof}

    \begin{figure}
        \centering
        \includegraphics[height=2.5in]{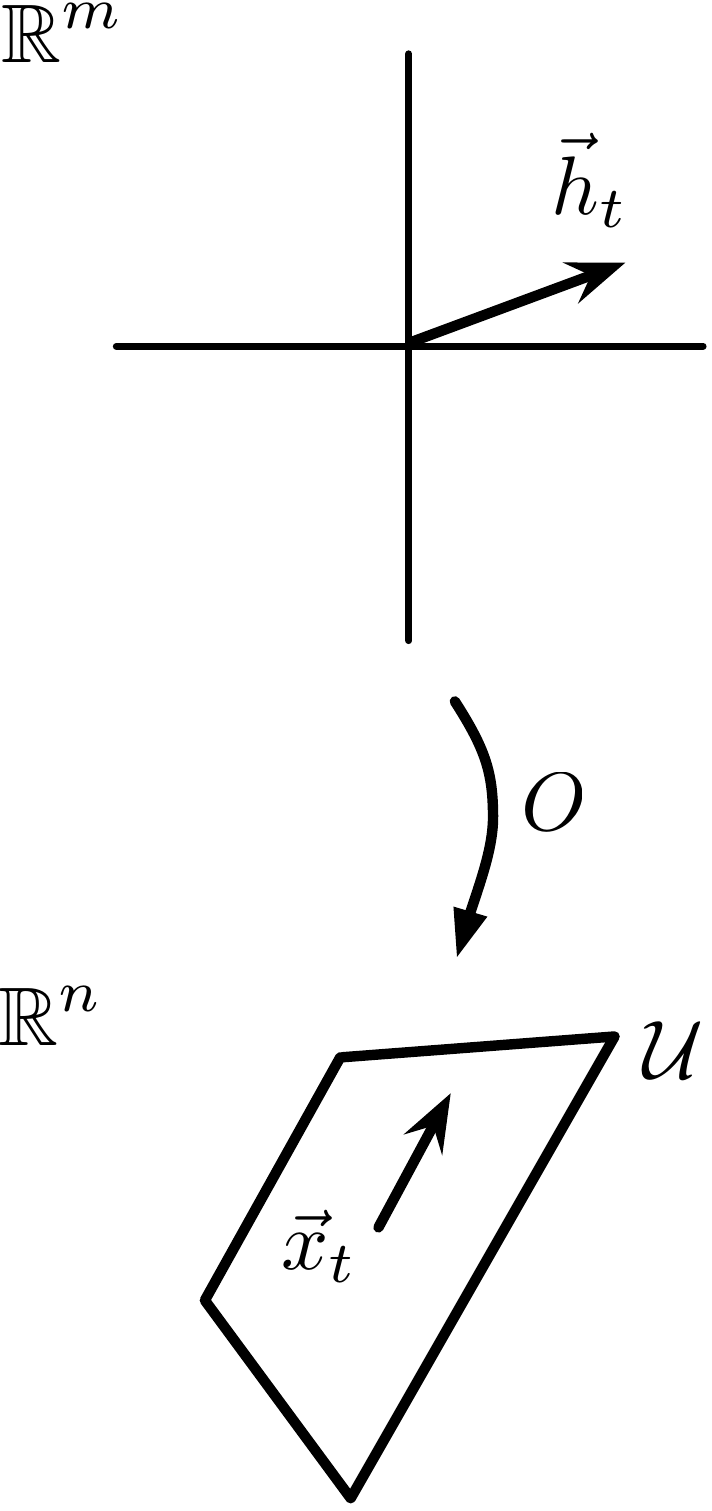}
        \caption{We can view $O$ as (the matrix of) a map from hidden state beliefs to output predictions (with respect to the standard bases). Not shown is the fact that both $\vh_t$ and $\vx_t$ lie in the simplices of $\Reals^m$ and $\Reals^n$, respectively, and that $O$ maps the simplex in $\Reals^m$ to (a subset of) the simplex in $\Reals^n$.}
        \label{fig:claim1}
    \end{figure}

    \begin{claim}
        The observation prediction subspace $\U$ satisfies $\U =
        \range(P_{2,1})$, where $[P_{2,1}]_{ij} := \Pr[x_2=i,x_1=j]$
        \label{claim:prediction-subspace}
    \end{claim}
    \begin{proof}
        We can write the joint distribution over $(x_1,x_2)$ as
        \begin{align}
            \Pr[x_2 = i, x_1 = j] &= \sum_{\bar{h}_1} \sum_{\bar{h}_2}
        \Pr[x_2=i,x_1=j,h_1=\bar{h}_1,h_2=\bar{h}_2]\\
        &= \sum_{\bar{h}_1} \Pr[h_1 = \bar{h}_1] \Pr[x_1 = j | h_1 = \bar{h}_1] \nonumber \\
        &\quad \cdot \sum_{\bar{h}_2} \Pr[h_2=\bar{h}_2 | h_1 = \bar{h}_1]  Pr[x_2 = i | h_2 = \bar{h}_2]
    \end{align}
        and we can write that sum as $P_{2,1} = O T \diagstar(\pi) O^\T$, where
        $\diagstar(\cdot)$ maps a vector to a diagonal matrix in the usual way. By
        our rank and positivity assumptions, we see that $P_{2,1}$ satisfies
        $\range(P_{2,1}) = \range(O) = \U$.
    \end{proof}

    We can directly estimate $P_{2,1}$ with empirical statistics, and as a
    consequence of Claim~\ref{claim:prediction-subspace} we can then get a
    basis for $\U$ by using an SVD of $P_{2,1}$. Note that we aren't getting an
    estimate of $O$ this way, but just its column space.

    \begin{claim}
        Given an observation $x_t = \bar{x}_t$, there is a linear map
        $\B_{\bar{x}_t}: \U \to \U$ such that
        \begin{align}
            \B_{\bar{x}_t}(\vx_t) = \alpha \vx_{t+1}
        \end{align}
        for some $\alpha = \alpha(\bar{x}_t,\vx_t)$, a scalar normalization
        factor chosen to ensure $1^\T \vx_{t+1} = 1$.
        \label{claim:update-map}
    \end{claim}
    \begin{proof}
        Following the usual recursive update for HMM forward messages, we have
        \begin{align}
            \Pr[h_{t+1} = i, x_{1:t} = \bar{x}_{1:t}]
            =
            \sum_{j \in [m]}
            \underbrace{\Pr[ h_{t+1} = i | h_t = j ]}_{[T]_{ij}}.
            \underbrace{\Pr[ x_t = \bar{x}_t | h_t = j ]}_{[O]_{\bar{x}_t j}} \nonumber \\
            \cdot \underbrace{\Pr[ h_t = j, x_{1:t-1} = \bar{x}_{1:t-1} ]}_{\propto [\vx_t]^O_j}.
        \end{align}
        Therefore we can write the map $\B_{\bar{x}_t}$ as an $m \times m$
        matrix relative to the basis for $\U$ given by the columns of $O$:
        \begin{align}
            [\B_{\bar{x}_t}]^O_O = T \diagstar(O_{\bar{x}_t:})
        \end{align}
        where $O_{k:}$ denotes the vector formed by the $k$th row of $O$. We
        require $1^\T \vx_{t+1}=1$, so we have $\alpha = O_{\bar{x}_t:}^\T
        \vx_t$.

        See Figure~\ref{fig:claim3}.
    \end{proof}

    \begin{figure}
        \centering
        \includegraphics[height=2.5in]{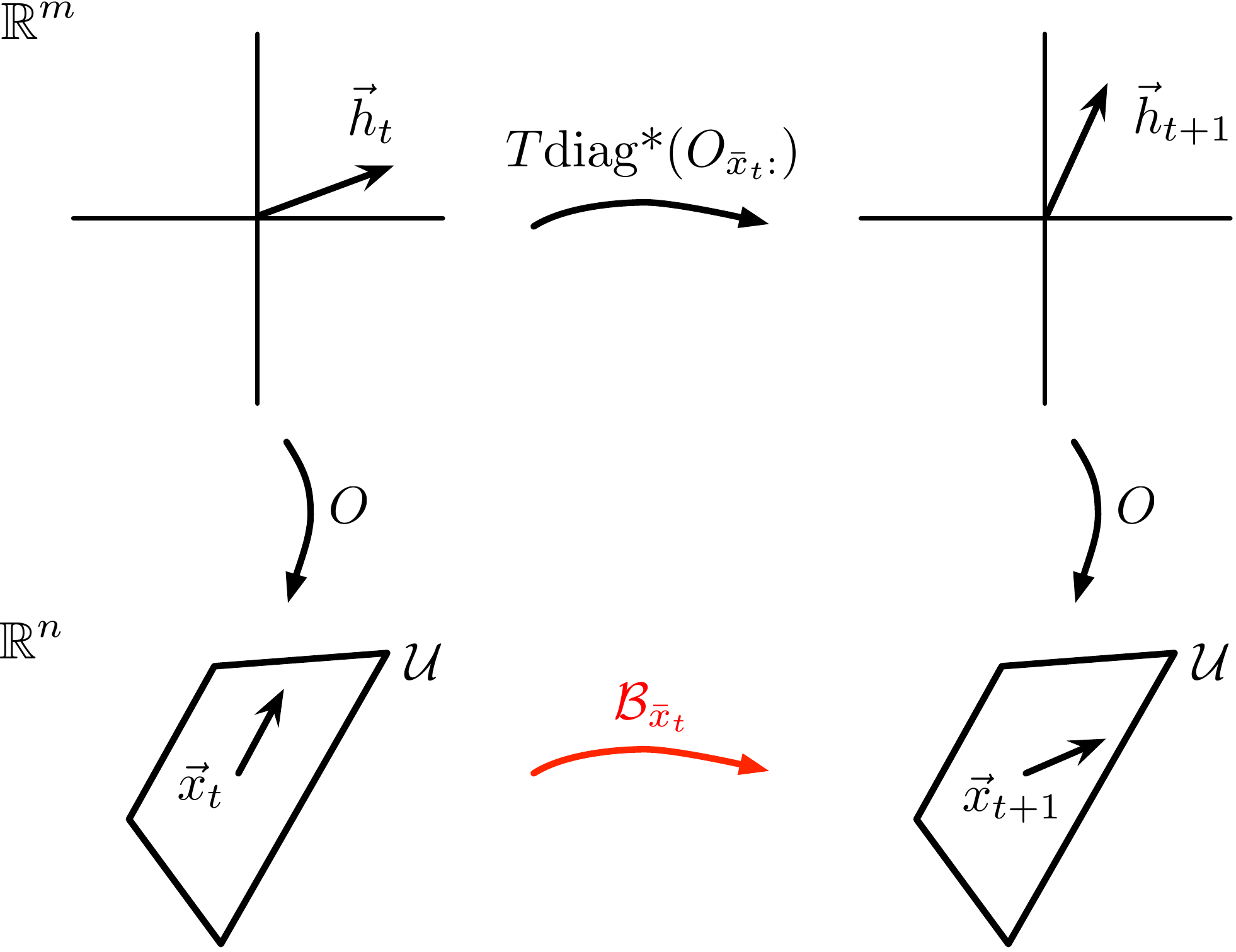}
        \caption{The matrix $T\diagstar(O_{\bar{x}_t :})$ is the matrix (relative to
            standard bases) of a linear map that updates hidden state beliefs
            given an observation $\bar{x}_t$, up to renormalization which the
            figure does not show. The linear map $\B_{\bar{x}_t}$ is the update
            for output predictions.}
        \label{fig:claim3}
    \end{figure}

    \begin{note}
        Renormalization always works because $T \diagstar(O_{\bar{x}_t :})$
        preserves the non-negative orthant of $\Reals^m$, and $O$ maps the simplex
        in $\Reals^m$ to (a subset of) the simplex in $\Reals^n$. The orthant
        preservation properties of these maps are an immediate consequence of the
        fact that the matrices (with respect to standard bases) are entry-wise
        non-negative. In fact, instead of tracking vectors, we should be tracking
        rays in the non-negative orthant.
    \end{note}

    For any new observation $x$, the map $\B_x$ implements the ``belief
    update'' on output predictions, up to a normalization factor which we can
    compute on the fly via $1^\T \vx_t = 1$. Note that we can also write $\B_x$
    as an $n \times n$ matrix relative to the standard basis of $\Reals^n$:
    \begin{align}
        [\B_x] = O T \diagstar(O_{x:}) O^\dagger
        \label{eq:B-standard-basis}
    \end{align}
    where $O^\dagger := (O^\T O)^{-1} O^\T$ is the pseudoinverse of $O$. Recall
    $\vx_t = O \vh_t$ and hence $\vh_t = O^\dagger \vx_t$.

    We would like to write $\B_x$ as a matrix without reference to the
    standard HMM parameters $(T,O,\pi)$, since we want to avoid learning them
    at all.

    \begin{claim}
        Let $U \in \Reals^{n \times m}$ be a matrix whose columns form an
        orthonormal basis for $\U$.
        We can write $\B_x$ as a matrix relative to the standard basis of $\Reals^n$ as
        \begin{align}
        [\B_x] = P_{3,x,1} U (U^\T P_{2,1} U)^{-1} U^\T
        \end{align}
        where
        \begin{align}
            [P_{3,x,1}]_{ij} := \Pr[x_3=i,x_2=x,x_1=j] \; \forall x \in [n].
        \end{align}
        \label{claim:update-map-observable}
    \end{claim}
    \begin{proof}
        We can express the matrix $P_{3,x,1}$ in a form similar to that for the
        matrix $[\B_x]$ in Equation~\eqref{eq:B-standard-basis}:
        \begin{align}
            P_{3,x,1} = \underbrace{O T \diagstar(O_{x:}) O^\dagger}_{[\B_x]} P_{2,1}.
        \end{align}
        Intuitively, we want to remove the $P_{2,1}$ on the right, since that
        would give us $[\B_x]$ in terms of quantities we can readily estimate,
        but we cannot form the inverse of $P_{2,1}$ because it is $n \times n$
        and has rank $m \leq n$.  However, $P_{2,1}$ has row and column space
        $\U$ (intuitively, its restriction to $\U$ is invertible), thus we can
        substitute
        \begin{align}
            P_{2,1} = U (U^\T P_{2,1} U) U^\T
        \end{align}
        to get
        \begin{align}
            P_{3,x,1} = O T \diagstar(O_{x:}) O^\dagger U (U^\T P_{2,1} U) U^\T
        \end{align}
        and hence
        \begin{align}
            P_{3,x,1} U (U^\T P_{2,1} U)^{-1} U^\T = O T \diagstar(O_{x:}) O^\dagger = [\B_x].
        \end{align}
    \end{proof}

    Because we can estimate each $P_{3,x,1}$ as well as $P_{2,1}$ from data by
    empirical statistics, and we can obtain a $U$ using an SVD, we can now estimate
    a representation of $\B_x$ from data using the expression in
    Claim~\ref{claim:update-map-observable}. We can also directly estimate
    $P_1$ from empirical statistics, where $[P_1]_i := \Pr[x_1 = i]$, and hence we can use
    these estimated quantities to recursively compute $\Pr[x_t|x_{1:t-1}]$ and
    $\Pr[x_{1:t-1}]$ given observations up to and including time $t-1$.

    Since $\dim(U)=m \leq n$, we can use the columns of $U$ as our basis for $\U$ to
    get a more economical coordinate representation of $\vx_t$ and $\B_x$ than
    in the standard basis of $\Reals^n$:

    \begin{definition}[HMM PSR Representation]
        For any fixed $U \in \Reals^{n \times m}$ with $\range(U)=\U$ and $U^\T U=I_{m
            \times m}$, we define the \emph{belief vector} at time $t$ by
        \begin{align}
            \vb_t := [\vx_t]^U = U^\T \vx_t \in \Reals^m
        \end{align}
        and in particular for $t=1$ we have
        \begin{align}
            \vec{b}_1 = U^\T O \pi.
        \end{align}
        For each possible observation $x \in [n]$, we define the matrix $B_x
        \in \Reals^{m \times m}$ by
        \begin{align}
            B_x := [\B_x]^U_U = (U^\T O) T \diagstar(O_{x:}) (U^\T O)^{-1}.
        \end{align}
        Finally, for normalization purposes, it is convenient to maintain the
        appropriate mapping of the ones (co-)vector, noting $1^\T O \vh = 1$ if
        and only if $1^\T \vh = 1$ because $1^\T O = 1^\T$:
        \begin{align}
            \vec{b}_\infty := [1]^U = U^\T 1 \in \Reals^m.
        \end{align}
    \end{definition}

    The box below summarizes the method for learning an HMM PSR representation
    from data and how to use an HMM PSR representation to perform some
    recursive inference computations.

    \begin{framed}
    \noindent \textbf{Learning}

    \begin{align*}
        \widehat{U} &= \text{ThinSVD}(\widehat{P}_{2,1})\\
        \widehat{b}_1 &= \widehat{U}^\T \widehat{P}_1\\
        \widehat{B}_x &= \widehat{U}^\T \widehat{P}_{3,x,1} (\widehat{U}^\T \widehat{P}_{2,1})^\dagger \; \forall x \in [n]\\
        \widehat{b}_\infty &= \widehat{U}^\T1 
    \end{align*}

    \paragraph{Inference}

    \begin{align*}
        \Pr[x_{1:t}] &= \vec{b}_\infty^\T B_{x_{t:1}} \vec{b}_1 &\text{sequence probability}\\
        \Pr[x_t | x_{1:t-1}] &= \vec{b}_\infty^\T B_{x_t} \vec{b}_t & \text{prediction}\\
        \vec{b}_{t+1} &= \frac{B_{x_t} \vec{b}_t}{\vec{b}_\infty^\T B_{x_t} \vec{b}_t} &\text{recursive update}
    \end{align*}
    \end{framed}
\end{document}